\documentclass[12pt]{article}
\usepackage{graphics}
\usepackage{amsmath,amsfonts,amssymb,amsthm}
\usepackage{epsfig,epstopdf,color}
\usepackage{enumitem}
\usepackage{float}
\usepackage[usenames,dvipsnames]{xcolor}
\usepackage{bbm}
\usepackage{geometry}
\usepackage{dsfont}

\newtheorem{theorem}{Theorem}[section]
\newtheorem{lemma}[theorem]{Lemma}

\def\E{\mathds{E}\,}
\def\Var{\mathds{V}ar}

\def\Pr{\mathds{P}}

\newcommand{\ignore}[1]{ }

\renewcommand{\le}{\leqslant}
\renewcommand{\ge}{\geqslant}

\newcommand{\brac}[1]{\left( #1 \right)}
\newcommand\bfrac[2]{\left(\frac{#1}{#2}\right)}
\newcommand\rai{\rightarrow \infty}

\def\b{\beta}
\def\d{\delta}

\def\e{\varepsilon}

\def\g{\gamma}

\def\th{\theta}

\def\l{\lambda}

\def\om{\omega}

\begin{document}

\author{Colin Cooper$^*$
\and Nan Kang$^*$\and Tomasz Radzik\thanks{Department of  Informatics, King's College London,  UK.}}
\title{A simple model of influence}
\maketitle

\begin{abstract}
We propose a simple model of influence in a network, based on edge density. In the model vertices (people) follow the opinion of the  group they belong to.  The opinion percolates  down from an active  vertex, the influencer, at the head of the group. Groups can merge, based on
interactions between influencers (i.e.,
interactions along `active edges' of the network),
so that the number of opinions is reduced. Eventually no active edges remain, and the groups and their opinions become static.

Our analysis is for $G(n,m)$ as $m$ increases from zero to $N={n \choose 2}$.
Initially every vertex is active, and finally $G$ is a clique, and with only one active vertex.
For $m\le N/\om$,
where $\om = \om(n)$ grows to infinity, but arbitrarily slowly,
we prove that the number of active vertices $a(m)$ is concentrated and we give w.h.p.\ results for this quantity. For larger values of $m$ our results give an upper bound on $\E a(m)$.

We  make an equivalent analysis for the same network when there are two types of influencers. Independent ones as described above, and stubborn vertices (dictators) who accept followers, but never follow. This leads to a reduction in the number of independent influencers as the network density increases. In the deterministic approximation (obtained by solving the deterministic recurrence corresponding to the formula for the expected change in one step), when
$m=cN$,
a single stubborn vertex reduces the number of influencers  by a factor of $\sqrt{1-c}$, i.e., from $a(m)$ to $(\sqrt{1-c})\,a(m)$. If the number of stubborn vertices tends to infinity slowly with $n$, then no independent influencers remain, even if $m=N/\om$.

Finally we analyse the size of the largest influence group which is of  order $(n/k) \log k$ when there are $k$ active vertices, and remark  that in the limit the size distribution of groups is equivalent to a continuous stick breaking process.

\end{abstract}


\section{Introduction}
We propose a simple model of influence in a network, based on edge density. In the model vertices (people) follow the opinion of the  group they belong to.  This opinion percolates  down from an active (or opinionated) vertex, the influencer, at the head of the group. Groups can merge, based on edges between influencers (active edges), so that the number of opinions is reduced. Eventually no active edges remain and the groups and their opinions become static.

The sociologist Robert Axelrod \cite{Axelrod} posed the question ``If people tend to become more alike in their beliefs, attitudes, and behavior when they interact, why do not all such differences eventually disappear?''.
This question was further studied by, for example, Flache et al. \cite{Fetc}, Moussaid et. al \cite{Mouss} who review various models of social interaction, generally based on some form of agency.

In our model the emergence of separate groups occurs naturally due to lack of active edges between influencers.  The exact composition of the groups and their influencing opinion being a stochastic outcome of the connectivity of individual vertices.

\paragraph{Joining Protocol.}

The process models how networks can partition into disjoint subgraphs which we call {\em fragments} based on following the opinion of a neighbour.  At any step, a fragment consists of a directed tree rooted at an active vertex (the influencer), edges pointing from follower vertices towards the root. This forms a simple model of influence where the vertices in a fragment follow the opinion of the vertex they point to, and hence that of the active root.

The process is carried out on a fixed underlying graph $G=(V,E)$. The basic u.a.r. process is as follows.
\begin{enumerate}
	\item Vertices are either active or passive. Initially all vertices of $V$ are active, and all fragments are individual vertices.
	\item
 \begin{enumerate}
     \item
     \emph{Vertex model}:
 An active vertex $u$ is chosen u.a.r. and contacts a random active neighbour~$v$.
 \item
 \emph{Edge model}:
 A directed edge $(v,u)$ between active vertices is chosen u.a.r. and the active vertex $u$ contacts its active neighbour~$v$.
 (Equivalently, an undirected edge is chosen u.a.r.
 and random of the two vertices contacts the other
 vertex).
 \end{enumerate}
	\item The contacted neighbour $v$ becomes passive.
	\item Vertex $v$ directs an edge to $u$ in the fragment graph. Vertex $v$ and its fragment $F(v)$ become part of the fragment $F(u)$ rooted at $u$.
	\item An active vertex is isolated if it has no edges to active neighbours in $G$.
The process ends when all active vertices are isolated.
\end{enumerate}
\paragraph{Summary of results.}
As an illustration of the process,
in this paper we make an analysis of the edge model for random graphs $G(n,m)$, providing the following results.
 \begin{itemize}
 \item Theorem \ref{thm:uar_Gnp_comp1} gives the w.h.p. number of fragments in $G(n,m)$ for $m \ll {n \choose 2}$, and an upper bound on the expected number for any $m$.  The results are supported by simulations
 which indicate that the upper bound in Theorem~\ref{thm:uar_Gnp_comp1} is the correct answer.
 \item Theorem \ref{Dicks}  gives the equivalent number of fragments in the presence of  stubborn vertices  (vertices who accept followers, but refuse to follow).
 \item The tail distribution of size of the largest fragment and its expected size are given in  Lemma \ref{lemma:spacings}.
  \end{itemize}

\section{Analysis for random graphs $G(n,m)$}

We suppose the underlying graph is a random graph $G(n,m)$\, and at each step, the absorbing vertex and the contacted neighbour are chosen by selecting uniformly at random an edge between two
active vertices (the edge model).

We work with a random permutation $\sigma$ of the edges of the complete graph, where $N={n \choose 2}$.
The edges of $\sigma$ are inspected in the permutation order. By revealing the first $m$ edges in the random permutation we choose a random graph $G(n,m)$. The order in which we reveal these
first $m$ edges and and their random directions
give a random execution of the joining protocol
on the chosen graph.

The u.a.r. process is equivalent to picking a random edge between active vertices
(by skipping the steps when one or both vertices of the chosen edge is not active).
One endpoint stays active and the other becomes passive. It doesn't matter which (since we are
interested in the number and sizes of fragments, but
not in their structure).
As none of the edges between active vertices have been inspected, the next edge is equally likely to be between any of them.

Let $A(m)$ be the set of active vertices obtained by running the process on $G(n,m)$.
Let $a(m)=|A(m)|$ be the number of active vertices after $m$ edges are exposed.

The following deterministic recurrence plays a central part in our analysis,
\begin{equation}\label{**}
a_{t+1}= a_t - \frac{a_t(a_t-1)}{2(N-t)}.
\end{equation}
We will show that if $t=o(N)$ then $\E a(t) \sim a_t$, and that $\E a(t) \le a_t$ always.
The solution to \eqref{**} is given in  the next lemma.
To maintain continuity of presentation,
the proof of the lemma is deferred to
the next section.

In what follows $\om$ is a generic variable which tends to infinity with $n$, but can do so arbitrarily slowly.

\begin{lemma} \label{DetLem}
$\;$
\begin{enumerate}[label=(\roman*)]
\item
For $t\le N/\om$, we have $a_t=c_t (1+\e_t)$ for $\e_t=O(t/(N-t))$ and
$c_t$ given by
 \begin{equation} \label{sparse}
  c_t = \frac{n^2}{n+t}.
\end{equation}
Thus if $t=o(N)$,
$a_t \sim c_t$.\vspace*{1ex}
\item
For $t \le N - \om$,
we have $a_t=b_t (1+\e_t)$ for $\e_t=O(1/\om)$ and
$b_t$ given by
\begin{equation}\label{denser}
b_t=\frac{1}{1-(1-1/n)\sqrt{1-t/N}}.
\end{equation}
\end{enumerate}
\end{lemma}
Our first result follows from this lemma.

\begin{theorem}\label{thm:uar_Gnp_comp1}
	Given a random graph $G(n,m)$\,, for $m\le N/\om$, where $\om \rai$, the number of components $a(m)$ generated by the opinion fragmentation process is  concentrated with expected value given by
\begin{equation}\label{Gnm}
\E a(m) \sim \frac{1}{1-(1-1/n)\sqrt{1-m/N}}.
\end{equation}
Moreover for any $m \le N(1-o(1))$,  $\E a(m)$ is upper bounded by the RHS of \eqref{Gnm}.

For $m\le N/\om$, the  expected number of active vertices $\E a(m)$ is well approximated by the simpler expression $\E a(m) \sim  \frac{n^2}{n+m}$ as given by \eqref{sparse}.
\end{theorem}

\begin{proof}
We add edges of a complete graph to an empty graph in random order, and analyse the expected change in the number of active vertices in one step. At the beginning all vertices are active and $a(0)=n$.

Let $a(t)$ be the total number of active vertices at step $t$\,. There are $N-t$ unexamined edges remaining after step $t$ as we add one edge per step, and  there are $\binom{a(t)}{2}$ many active edges left after step $t$\,. Therefore, the probability of choosing an active edge at step $t+1$ is $\binom{a(t)}{2}/(N-t)$\,, and we lose one active vertex for each  active edge added.
Thus,
	\begin{align} \label{eq:uar_Gnp_activeNodes}
		\E[a(t+1) \mid a(t)] = a(t) - \frac{\binom{a(t)}{2}}{N-t} \,.
	\end{align}
 The function $x^2$ is convex so $\E (a(t))^2 \ge (\E a(t))^2$. Thus   the solution $a_t$ of the recurrence \eqref{**}
 gives an upper bound on $\E a(t)$. 

 On the other hand, if $a(t)$
is concentrated, then $\E (a(t))^2 \sim (\E a(t))^2$ in which case $\E a(t) \sim a_t$ as in \eqref{**}. This is easy up to $t\le n^{4/3}/\om$.
Using a edge exposure martingale, the value of $a(t)$ can only change by zero or one at any step, so
\begin{equation}\label{jcfkiew0a}
\Pr(|a(t)- \E a(t)| \ge \l) \le \exp \brac{- \frac{\l^2}{2t}}.
\end{equation}
For $t\le n^{4/3}/\om$, choose $\l= \sqrt{\om t}$
to get $o(1)$ on the RHS in~\eqref{jcfkiew0a}
and $\l = n^{2/3} = o(a_t)$.
Assuming concentration of $a(t)$,  $\E (a(t))^2$ on the RHS of \eqref{eq:uar_Gnp_activeNodes} can be replaced  by
$(1+o(1)) (\E a(t))^2$. This  allows us to use recurrence \eqref{**} to analyse the recurrence  \eqref{eq:uar_Gnp_activeNodes} for $\E a(t)$.

From $t \ge n^{4/3}$ onward, mostly nothing happens at any step and the standard Azuma-Hoeffding inequality approach stops working.
As $a(t)$ is a supermartingale ($\E a(t+1) \le a(t)$), we can use Freedman's inequality, which we paraphrase from \cite{Easy}.

{\sc Freedman's Inequality \cite{Free}. \;}{\em  Suppose $Y_0,Y_1,...$ is a supermartingale such that
$Y_j-Y_{j-1} \le C$ for a positive constant $C$ and
all $j$. Let $V_m=\sum_{k \le m} \Var [(Y_{k+1}-Y_k) \mid {\cal F}_k]$. Then for $\l, b >0$
\begin{equation}\label{Freem}
\Pr(\exists m: V_m \le b \text{ and } Y_m- \E Y_m \ge \l) \le \exp\brac{-\frac{\l^2}{2(b+C\l)}}.
\end{equation}
}\\
In our case $Y_m=a(m+s)$ given the value of $a(s)$, and $C=1$ by \eqref{eq:uar_Gnp_activeNodes}.

\begin{figure}[H]
	\centering
	\includegraphics[scale=.65]{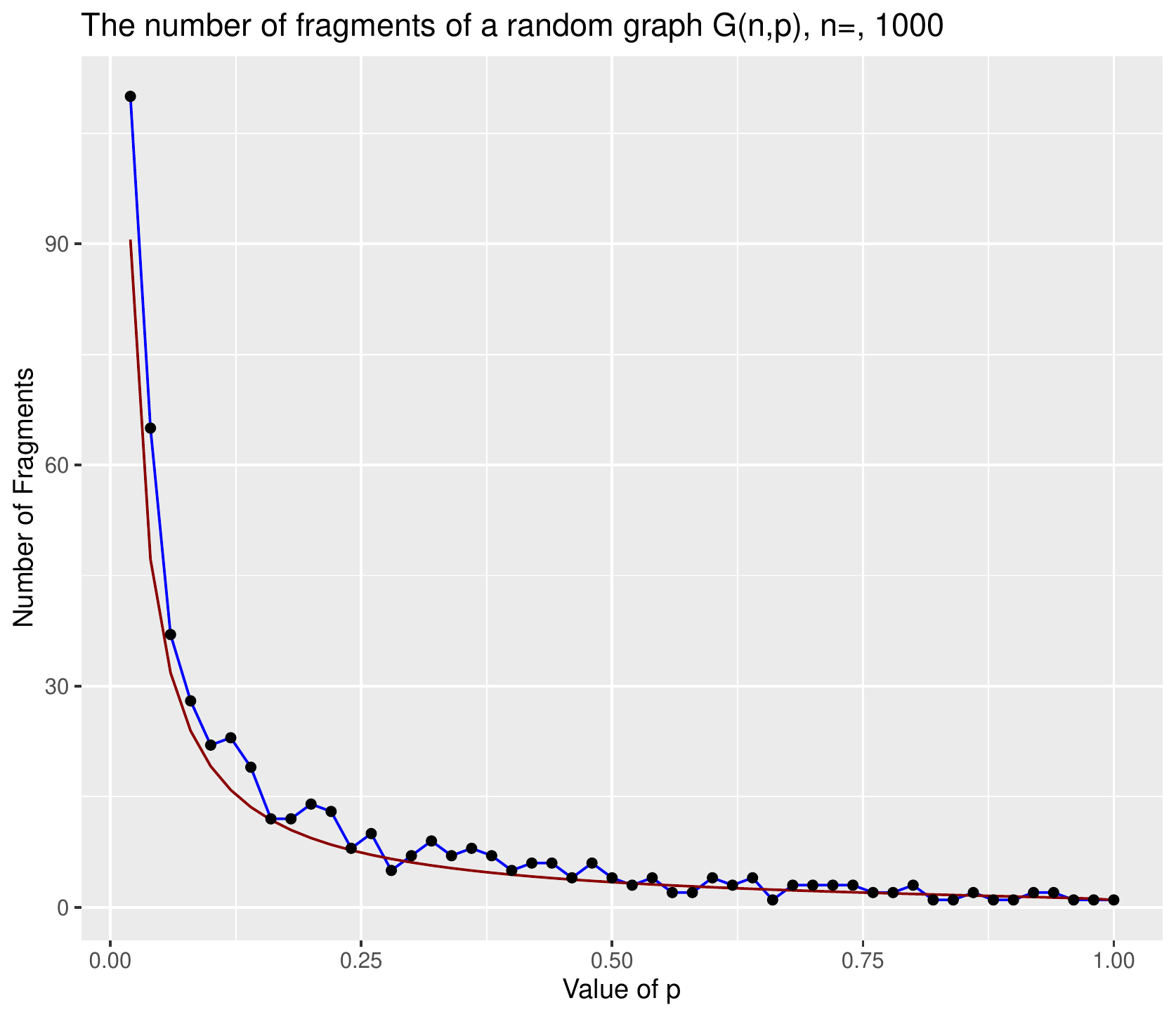}
	\caption{Simulation of the number of active vertices. The simulation is based  on $G(n,p)$. For large $m=Np$, $G(n,p) \sim G(n,m)$ so the results are equivalent.
 The blue plot is the simulation, with values of $p$ interpolated at $0.02$, the first entry being $p=0.02$. The red curve is  \eqref{denser} giving $b_m$ evaluated at $p=m/N$. } \label{Fig1}
 \end{figure}

Let $t_1=n^{4/3}/\om=n^{1+1/3}/\om$ and $t_i=n^{1+1/3+\cdots+2^{i-1}/{3^i}}/\om=n^{1+\b_i}/\om$, where
$\b_i=1-2^i/3^i$ and $\om$ may vary. The inductive assumption is that
\[
a(t_i) \sim \frac{n^2}{n+t_i} \sim \om n^{1-\b_i},
\]
see Lemma \ref{DetLem}. As $a(t)$ is monotone non-increasing it follows  for $t_i\le t\le t_{i+1}$, and $t_i=o(N)$, that
\[
\frac{{a(t) \choose 2}}{N-t} \le \frac{a(t_i)^2}{n^2} (1+o(1)) \sim \frac{\om^2 n^{2-2\b_i}}{n^2}=\frac{\om^2}{n^{2\b_i}}.
\]
As $\Var(a(t+1)-a(t))\le (1+o(1))a(t_i)^2/n^2$ we have that
\[
b_i=\sum_{t=t_i}^{t_{i+1}}\Var(a(t+1)-a(t))\le (1+o(1))\frac{t_{i+1}}{n^{2\b_i}}=\om n^{1-\b_i+2^i/3^{i+1}}.
\]
Thus using \eqref{Freem},
\begin{align*}
\Pr(a(t_{i+1} & - \E a(t_{i+1}) \ge \om^{3/4} n^{1-\b_{i+1}})\\
\le&\;
\exp\brac{-\frac 12 \om^{6/4} \frac{n^{2-2\b_{i+1}}}{\om n^{1-\b_i+2^i/3^{i+1}}(1+o(1))}}\\
=&\; \exp\brac{-\frac 12 \om^{1/2} n^{1-\b_{i+1}-2^{i+1}/3^{i+1}}}
\; =\; \exp\brac{-\frac 12 \om^{1/2}}.
\end{align*}
The last line follows because $ \b_{i+1}=1-2^{i+1}/{3^{i+1}}$. For simplicity let $\om= C \log^2 n$ for some large constant $C$.
As  $i \rai$,  $\b_i$ tends to one, and we have that w.h.p. $a(t) \sim \E a(t)$ for any $t=o(N)=N/\om$ say.
Hence $\E(a(t)^2) \le (1+o(1)) (\E a(t))^2$. From an earlier part of this theorem, $\E (a(t))^2 \ge (\E a(t))^2$.
This completes the proof of the theorem.
 \end{proof}

\ignore{

{\bf Remarks.}

We are currently not able to prove asymptotic equality between $\E a(t)$ and $b_t$ for values of $t$ larger than $t=N/\om$.  As $a(t)$ decreases, the wait for an active edge increases; most of the time nothing is happening.
Eventually  there is no possibility of concentration around $\E a(t)$. Indeed when $t=cN$, $c\le 1$ constant,  $\E a(t) \le 1/(1-\sqrt{1-c})$.

If the recurrence for $\E a(t)$ in \eqref{eq:uar_Gnp_activeNodes} is replaced by the deterministic recurrence \eqref{**}
then the solution \eqref{denser} holds for $t \in [1, N-\om]$. So the result $a(t) \sim b_t$ holds in 'mean field' for that range, and $\E a(t) \le b_t (1+o(1)$ always.
}

Simulation results (see Figure \ref{Fig1}), suggest that $\E a(t)$ should continue to track  $b_t$ of \eqref{denser} throughout.

\section{Proof of Lemma \ref{DetLem}. }

To solve \eqref{eq:uar_Gnp_activeNodes},
the first step is to solve the equivalent deterministic recurrence \eqref{**}, i.e.,
\begin{equation}\label{**1}
a_{t+1}= a_t\brac{1 - \frac{a_t-1}{2(N-t)}}.
\end{equation}
An approximate solution can be obtained  by
replacing $a_t$ in this recurrence by a differential equation in $b(t)$. The initial condition $b(0)=a_0=n$ gives
	\begin{align}\label{b(t)}
		\frac{ d\, b(t)}{d t} = - \frac{\binom{b(t)}{2}}{N-t} \qquad
		\implies \qquad  b(t) = \frac{1}{ 1-\left(1- \frac{1}{n} \right)\sqrt{1-\frac{t}{N}}}.
	\end{align}
We now prove Lemma \ref{DetLem}, restating it
below for convenience.
\begin{lemma} 
$\;$
\begin{enumerate}[label=(\roman*)]
\item
Let $c_t$ be given by,
 \begin{equation}\nonumber 
  c_t = \frac{n^2}{n+t}.
\end{equation}
For $t = N/\om$, we have $a_t=c_t (1+\e_t)$ for
$\e_t=O(t/(N-t))$. Thus if $t=o(N)$,
$a_t \sim c_t$.\vspace*{1ex}
\item
For $t \le N-\om$, where $\om \rai$,
we have $a_t=b_t (1+\e_t)$ for $\e_t=O(1/\om)$ and
$b_t$ given by
\begin{equation}\nonumber 
b_t=\frac{1}{1-(1-1/n)\sqrt{1-t/N}}.
\end{equation}
\end{enumerate}
\end{lemma}
\begin{proof}
\emph{(i)} We define $\e_t$ so that $a_t=c_t (1+\e_t)$, and show by induction that for $0 \le t \le N/\om$,
we have $-2t/(N-t)\le \e_t \le 0$, starting from $a_0=c_0=n$ and $\e_0 = 0$.
We take an arbitrary $0 \le t \le N/\om$, and evaluate the recurrence \eqref{**1} for $a_{t+1}$,
assuming inductively that $a_t=c_t(1+\g t/(N-t))$ for some $\g \in [-2, 0]$.
\begin{align*}
\frac{a_t-1}{2(N-t)}
= \; & \; \frac{\frac{n^2}{n+t}\brac{1+\frac{\g t}{N-t}}-1}{2(N-t)}
\; = \; \frac{n^2- n - t}{2(n+t)(N-t)} \; + \; \frac{\g t n^2}{2(n+t)(N-t)^2} \\
\; = \; & \; \frac{2(N-t) + t}{2(n+t)(N-t)} \; + \; \frac{\g t n^2}{2(n+t)(N-t)^2} \\
\; = \; & \; \frac{1}{n+t} \; + \; \frac{t}{2(n+t)(N-t)} \; + \; \frac{\g t n^2}{2(n+t)(N-t)^2}.
\end{align*}
Thus
\begin{align*} \nonumber
\frac{a_{t+1}}{c_{t+1}}
= \; & \; \frac{c_t}{c_{t+1}}\brac{1+\frac{\g t}{N-t}}\brac{1 - \frac{1}{n+t} - \frac{t}{2(n+t)(N-t)} - \frac{\g t n^2}{2(n+t)(N-t)^2}}\\
\nonumber
= \; & \; \brac{1+\frac{1}{n+t}}\brac{1+\frac{\g t}{N-t}}\\
& \; \times \brac{1 - \frac{1}{n+t} - \frac{t}{2(n+t)(N-t)} - \frac{\g t n^2}{2(n+t)(N-t)^2}}\\
= \; & \; 1 +\frac{\g t}{N-t} - \frac{t}{2(n+t)(N-t)} - \frac{\g t n^2}{2(n+t)(N-t)^2} + \d
\; = \; 1 + \xi +\d,
\nonumber
\end{align*}
where $-1/(N-t) \le \d \le 1/(2(N-t))$, by
inspecting the terms contributing to $\d$ and
using the assumption that $t\le N/\om$.
Now we have, recalling that $-2 \le \g \le 0$,
\[
\xi = \frac{\g t}{N-t}\brac{1 - \frac{n^2}{2(n+t)(N-t)}} - \frac{t}{2(n+t)(N-t)}
\le - \frac{1}{2(N-t)}
\]
and
\begin{align*}
\xi
& \ge -\frac{2t}{N-t} - \frac{t}{2(n+t)(N-t)}\\
& \ge -\frac{2(t+1)}{N-t-1}
+ \frac{2N}{(N-t-1)(N-t)}
- \frac{t}{2(n+t)(N-t)}\\
& \ge -\frac{2(t+1)}{N-t-1}
+ \frac{1}{N-t}\brac{\frac{2N}{N-t-1} - \frac{t}{2(n+t)}}
\ge -\frac{2(t+1)}{N-t-1}
+ \frac{1}{N-t}.
\end{align*}
The bounds on $\xi$ and $\d$ imply that
${a_{t+1}}/{c_{t+1}} = 1 + \g'{(t+1)}/{(N-t-1)}$, for some
$-2 \le \g' \le 0$.

\vspace{0.1in}\noindent
{\emph{(ii)}} We note firstly  that $c_t \ge b_t$ and establish that for $t=o(N)$, $c_t \sim b_t$.
Let $t_1 \le N/\om$. Using  $\sqrt{1-x}= 1-x/2-x^2/4-O(x^3)$, it can be checked that
\[
\frac{1}{1-(1-1/n)\sqrt{1-t_1/N}} =(1+\d) \frac{n^2}{n+t_1},
\]
where $\d=O(t_1^2/n^3(n+t_1))=O(1/\om^2)$.  Thus $a_{t_1}=b_{t_1}(1+O(1/\om))$.

Let $\th=(n-1)/n$.
Assume $N-t\ge \om \rai$. Then
\begin{align*}
\frac{1}{b_{t+1}} = & 1-\th \sqrt{1-\frac{t+1}{N}}=  1-\th \sqrt{1-\frac{t}{N}}\sqrt{1-\frac{1}{N-t}}\\
=&1-\th \sqrt{1-\frac{t}{N}}\brac{1-\frac{1}{2(N-t)}-\frac{1}{4(N-t)^2}(1+O(1/\om))},
\end{align*}
and so
\begin{align*}
\frac{b_t}{b_{t+1}}=1+ \frac{\th \sqrt{1-\frac{t}{N}}}{2(N-t)(1-\th\sqrt{1-\frac{t}{N}})}+
\frac{\th \sqrt{1-\frac{t}{N}}}{4(N-t)^2(1-\th\sqrt{1-\frac{t}{N}})}(1+o(1)).
\end{align*}
Let
\[
\l=\frac{\th \sqrt{1-\frac{t}{N}}}{2(N-t)(1-\th\sqrt{1-\frac{t}{N}})},
\]
then also $\l=(b_t-1)/(2(N-t))$.
Thus
\begin{flalign*}
a_{t+1}=&a_t \brac{1-\frac{a_t-1}{2(N-t)}}\\
=& b_t(1+\e_t)\brac{1-\frac{(b_t-1)(1+\e_t)+\e_t}{2(N-t)}}\\
=& b_{t+1}\brac{1+\l+\frac{\l}{2(N-t)}(1+o(1))}(1+\e_t)\brac{1-\l(1+\e_t)-\frac{\e_t}{2(N-t)}}\\
=&b_{t+1}\brac{1+\e_t-\l^2-\frac{\l}{2(N-t)}+O\brac{\e_t\; \brac{\l+\frac{1}{N-t}}}}.
\end{flalign*}
It follows from $\th \sqrt{1-x} < \sqrt{1-x} \le 1-x/2$, that $b_x \le 2/x$. Thus $b_t \le 2N/t$,
$\l \le N/(t(N-t))$ and
\[
\l^2 \le \frac{N}{t^2(N-t)}, \qquad \qquad \frac{\l}{N-t} \le \frac{N}{t(N-t)^2}.
\]
Finally with $t_1=N/\om$ as above
\[
|\e_t| \le |\e_{t_1}|+ \sum_{t_1}^t \l^2+\frac{\l}{N-t} \sim  \int_{t_1}^t \frac N{t^2(N-t)} +\frac{N}{t(N-t)^2}
=\int F(t) dt.
\]
Denote $t_1=c_1N$ and $t_2=c_2N$ where $t_2 \le N-\om$. Then
\begin{align*}
\int F(t) dt=& \int \frac{1}{t^2}+ \frac{2}{Nt}-\frac{2}{N(N-t)} +\frac{3}{(N-t)^2}\\
=&\left[-\frac 1t +\frac{2}{N}\log t +\frac{2}{N}\log (N-t) + \frac{3}{N-t}\right]^{\!Nc_2}_{\!Nc_1}\\
\le&\frac 1N \brac{  \frac 1c_1+ 2\log \frac{c_2}{c_1}+2 \log
\frac{1-c_2}{1-c_1} +\frac 1{1-c_2}},
\end{align*}
and thus if $c_2=1-\om/N$, from the last term,
\[
|\e_{t_2}| \le O\bfrac{1}{\om}.
\]

\end{proof}

\section{The effect of stubborn vertices.}

A vertex is  stubborn (intransigent, autocratic, dictatorial) if it holds fixed views, and  although happy to accept followers, it refuses to   follow the views of others. Typical examples include news networks, politicians and some cultural or religious groups. Stubborn vertices can only be root vertices.

We note that voting  in distributed systems in the presence of stubborn agents has been extensively studied
see e.g., \cite{Mucko}, \cite{Nico}, \cite{Yildiz} and references therein.

The effect of stubborn vertices on the number of other active vertices in the network depends on the edge density, as is illustrated by the next theorem. Let $a_k(t)$ be the number of {\em active independent vertices} at step $t$ in the presence of $k \ge 0$ stubborn vertices. As the stubborn vertices are never absorbed, the total number of roots is $a_k(t)+k$.

Let $\b=2k-1$, and
\begin{equation}\label{b1(t)}
b_k(t) \sim \b \;\frac{ \bfrac{n}{n+\b} \brac{1-\frac tN}^{\b/2}}{1- \bfrac{n}{n+\b} \brac{1-\frac tN}^{\b/2}}.
\end{equation}
Essentially we solve the  deterministic recurrence equivalent to \eqref{**} to obtain \eqref{b1(t)}, and argue
by concentration, convexity and super-martingale properties that $b_1(t)$ is the asymptotic solution ($t=o(N)$)
or an effective upper bound ($t \le N$). Due to space limitations the proof is only given in outline.

\begin{theorem}\label{Dicks}
(i) {\bf One stubborn vertex.}
Let  $N={n \choose 2}+n$, then  provided $t\le N/\om$,  the number of  independent active vertices
$a_1(t) \sim a_0(t) \sim b_t$, w.h.p., where $b_t$ is the solution to \eqref{**} as given by \eqref{denser}.
If $t=cN$ then $\E a_1(cN) \le b_1(cN) \le (\sqrt{1-c})\, b_t$.

(ii) {\bf A constant number $k$ of stubborn vertices.} Let $k \ge 1$ be integer, and $N=k+{n \choose 2}$. If $k$ is constant, and $t\le N/\om$ then
$a_k(t) \sim  b_k(t)$ w.h.p., and for any $t \le N$,  $\E a_k(t) \le b_k(t) (1+o(1)$.

(iii) {\bf The number $k$ of stubborn vertices is unbounded.} If $t=N/\om$ and $\om/k \rightarrow 0$, then w.h.p no independent active vertices are left by step $t$.
\end{theorem}
\begin{proof}
To formulate the model, we note that, at the end of step $t$ there are $k a_k(t)$ edges between stubborn and independent active vertices. Writing $a=a_k(t)$ we extend \eqref{eq:uar_Gnp_activeNodes} with $N={n \choose 2}+k$ to
\begin{equation}\label{Dde}
\E a_k(t+1)= a_k(t) -\frac{ka_k}{N-t} - \frac{a_k(a_k-1)}{2(N-t)}.
\end{equation}
Solving the equivalent differential equation we obtain \eqref{b1(t)}

In case (iii), let $t=N/\om$ and $k=\l \om$ where $\l \rai$, but $k=o(n)$. Then
\[
(1-t/N)^{\b/2} \sim e^{-t(2k-1)/2N}= e^{-(k-1/2)/\om} \sim e^{-\l} =o(1).
\]
Thus \eqref{b1(t)} tends to $b_1(t) \sim o(1)/(1-o(1))$ and the result follows.
\end{proof}
We remark that if the network is sparse ($c=o(1)$),  and there are only a few stubborn vertices, these will  have little effect. However, if the network is dense ($c$ is a positive constant),
there are fewer independent active vertices, even if $k$ is constant.  On the other hand if $k \rai$
even in sparse networks where $t=N/\om$, the number of independent active vertices can tend to zero.
This indicates in a simplistic way the effect of edge density (increasing connectivity)  in social networks
 on the formation of independent opinions in the presence of vertices with fixed views. It also indicates that even in sparse networks,  a large number of stubborn vertices can lead to the suppression of independent  opinion formation.

Figure \ref{Fig2}, illustrates the above Theorem.
The plots show the number of active vertices in the presence of stubborn vertices (dictators).
The number $k$ of stubborn vertices is equal to $1$ in
the left hand plot and $5$ in the right hand  plot.
The plots are based on $G(n,p)$, for $n=1000$ and $p \ge 0.1$.
The upper curve in the right hand figure is $b_k(t)+k$, the total number of active vertices.
The middle curve is $b_t$ from  \eqref{denser}. The  simulation plot marked by $+$ symbols is the final number of active vertices in a system without stubborn vertices, as in Figure \ref{Fig1}. The lower curve is $b_k(t)$, and the associated simulation is the number of independent active vertices in the presence of dictators.

 In the left hand plot for $k=1$, the curves $k+b_k(t)$ and $b_t$ as given by \eqref{denser} are effectively identical, so a distinct upper curve is missing.
 The lower curve is $b_1(t)$, and its associated plot is the number of independent active vertices in the presence of stubborn vertices.

\begin{figure}[H]
 \centering
  \includegraphics[width=.48\linewidth]{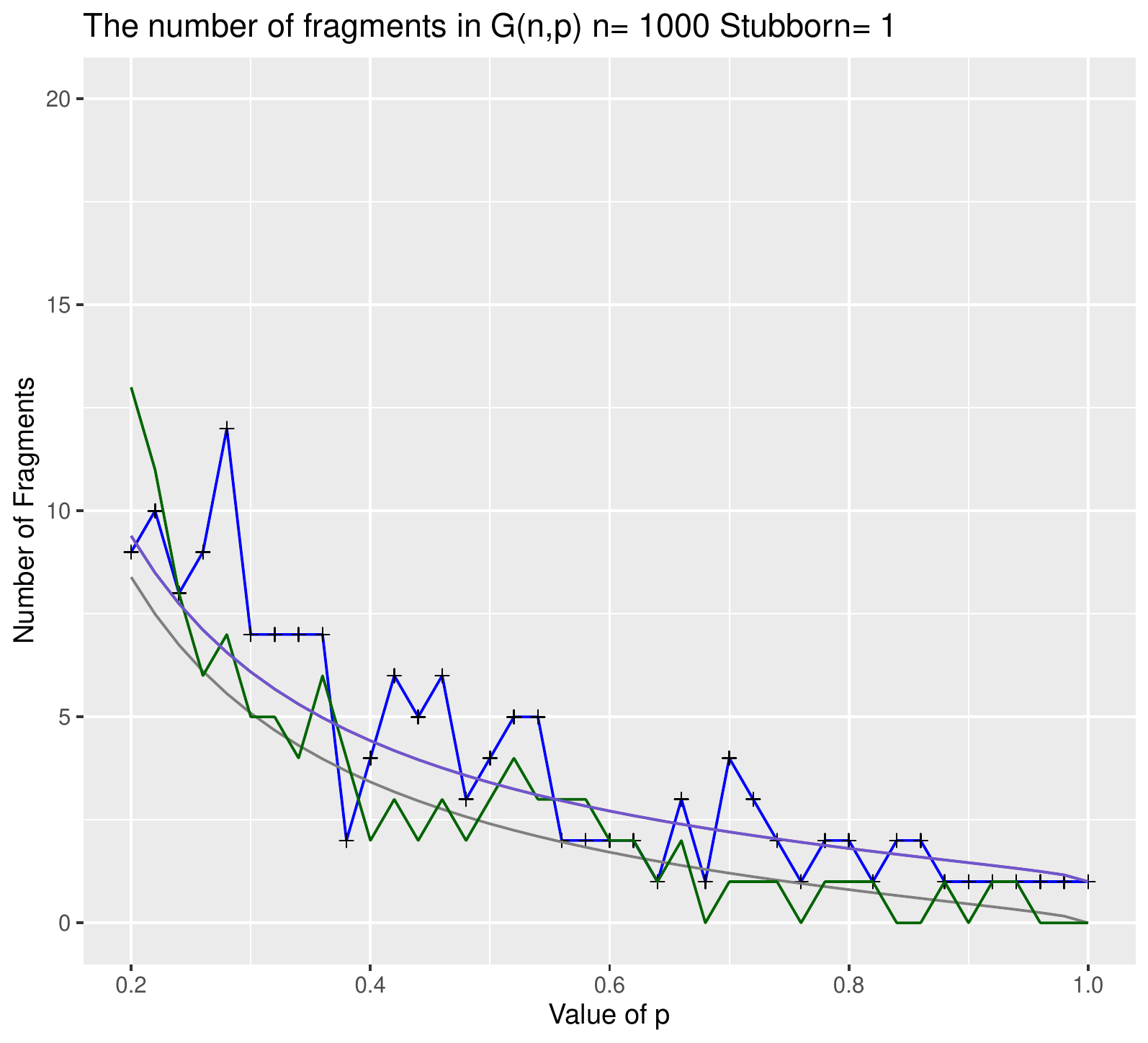}   %
\quad \includegraphics[width=.48\linewidth]{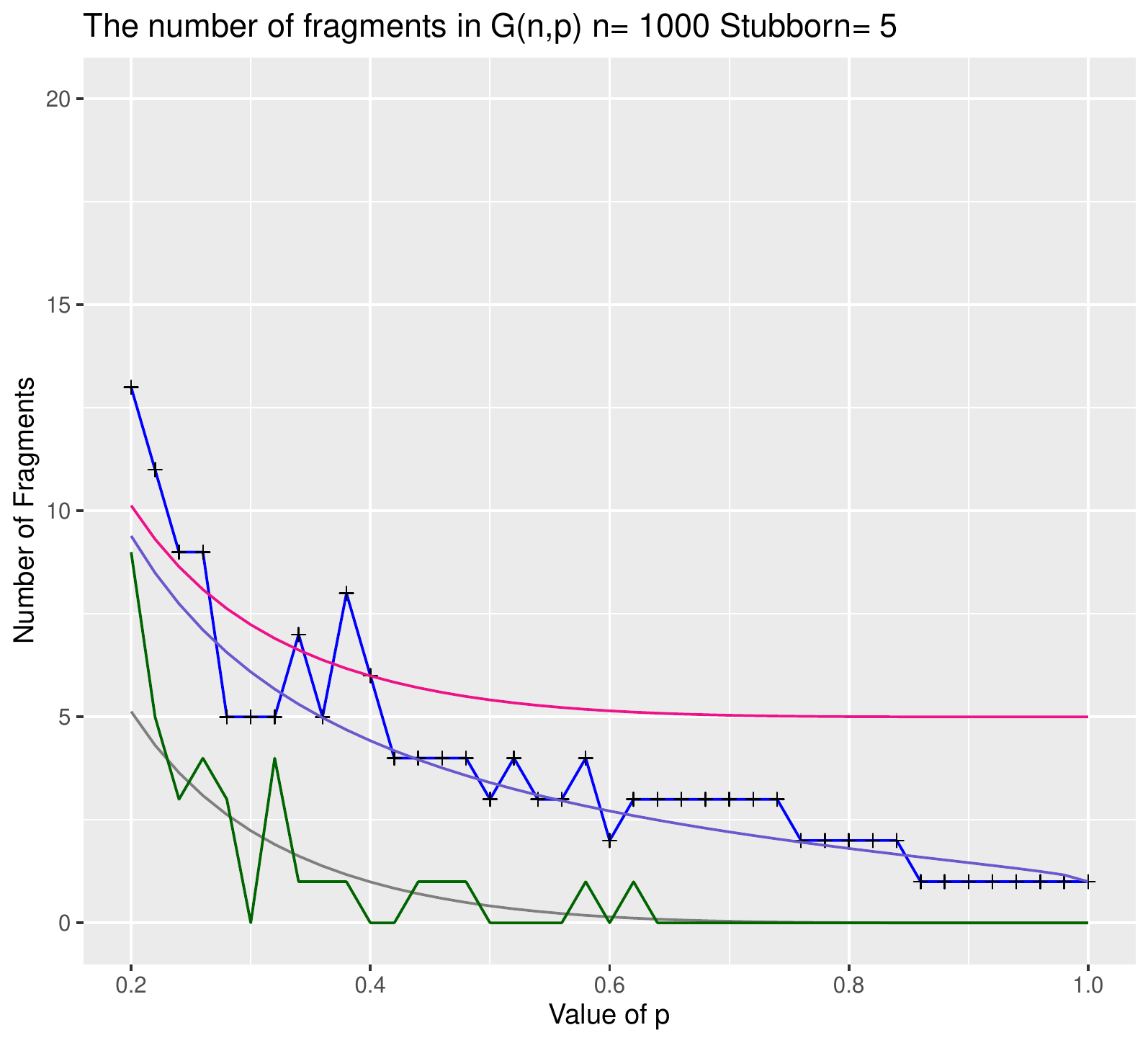}

	\caption{The number of active vertices in the presence of stubborn vertices (dictators). The left hand plot is for $k=1$ stubborn vertices, and the right hand  plot for $k=5$ stubborn vertices. The plots are based on $G(n,p)$, for $n=1000$ and $p \ge 0.1$.
 } \label{Fig2}
\end{figure}

\ignore{
\paragraph{As the number of stubborn vertices $k$ becomes large.}
Let $k \rai$, but $k=o(n)$. Then $\b \sim 2k$, and
\[
b_k(t) = 2k \frac{e^{-k t/N}}{1-e^{-kt/N}} (1+o(k/n)).
\]+++
Let $t^*=(N/k)(\log k+\om)$, then
\[
b_k(t^*) \sim (1+o(1)) 2 e^{-\om},
\]
giving an  extinction time for the independent vertices $a_k(t)$. This is approximated by $p^*=(\log k) /k*$ in the plot.
}

\section{The largest fragment  in $G(n,m)$}

Let $F_{(1)}(m)$ denote the size of the largest fragment in $G(n,m)$.
The value of $F_{(1)}(m)$, the number of followers of the dominant influencer, (we assume the influencer follows themself), will depend on the number of active vertices $a(m)$. As both $a(m)$ and $F_{(1)}$ are random variables, it is easier to fix $a(m)=k$, and
 study $\E F_{(1)}=\E( F_{(1)} \mid k)$  for a given value $k$. In the limit as $n \rai$, $\E F_{(1)}$ converges to a continuous process known as stick breaking.

The first step is to describe a consistent discrete model. This can be done in several ways, as a multivariate Polya urn, as the placement of $n-k$ unlabelled balls into $k$ labelled boxes, or as randomly choosing $k-1$ distinct vertices from $\{2,...,n\}$ on the path $1,2,...,n$. The latter corresponds to the limiting stick breaking process.

\paragraph{ Looking backwards: A Polya urn process.}
If we stop the process when there are exactly $k$ active vertices for the first time, then at the previous step there were $k+1$ active vertices. Let the $k$ active vertices be $A_k=\{v_1,...,v_k\}$, and  let
$A_{k+1}= \{v_1,...,v_k, b\}$ be the active vertices at the previous step, where the vertex $b$ was absorbed. As the edges $bv_1,...bv_k$ are equiprobable, the probability $b$ was absorbed by $v_i$ is $1/k$.

Working backwards from $k$ to $n$ is equivalent to a $k$-coloured Polya urn, in which at any step a ball is chosen at random and replaced with 2 balls of the same colour. At the first step backwards any one  of the colours $1,...,k$ is chosen and replaced with 2 balls of the chosen colour (say colour $i$).
This is equivalent to the event that vertex $b$ attaches to the active vertex $v_i$.

Starting with $k$ different coloured balls and working backwards for $s$ steps is equivalent to placing $s$ unlabelled balls into $k$ cells. Thus any vector of occupancies $(s_1,...,s_k)$ with $s_1+\cdots+s_k=s$ is equivalent to a final number of balls $(s_1+1,s_2+1,...,s_k+1)$; which is the sizes of the fragments at this point. The number of distinguishable solutions to $(s_1,...,s_k)$ with $s_1+\cdots+s_k=s$ is given by
\[
A_{s,k}={s+k-1 \choose k-1}.
\]
As finally  the number of vertices is $s+k=n$, if the process stops with $k$ distinct fragments, there are $N(k)={n-1 \choose k-1}$ ways to partition the vertices among the fragments, all partitions being equiprobable.

An illustration of the balls into cells process  is given by the stars and bars model in Feller \cite{FeI}. Quoting from page 37, \lq We use the artifice of representing the $k$ cells by the spaces between $k+1$ bars and the balls by stars. Thus $|***|*|\;|\;|\;|****|$ is used as a symbol for a distribution of $s=8$ balls into $k=6$  cells with occupancy numbers $3,1,0,0,0,4$. Such  a symbol necessarily starts and ends with a bar but the remaining $k-1$ bars and $s$ stars can appear in an arbitrary order.\rq

That this is equivalent to the above Polya urn model can be deduced from $A_{s+1,k}/A_{s,k}=(s+k)/(s+1)$.
The numerator is the number of positions for the extra ball. Picking a left hand bar corresponds to picking one of the $k$ root vertices. The denominator is the number of ways to de-identify the extra ball; being the number of symbols (urn occupancies) which map to the new occupancy.

The sizes of the fragments can also be viewed as follows. Consider the path $1,2,3,...,n$ with the first fragment starting at vertex 1, the left hand bar.
The choice of  $k-1$ remaining start positions (internal bars) from the vertices $2,3,...,n$ divides the path into $k$ pieces whose lengths are the fragment sizes.
Taking the limit as $n \rai$ and re-scaling the path length to 1, we obtain the limiting process, known as {\em stick breaking}.

\paragraph{Limiting process: Stick breaking.}
The continuous limit as $n \rai$ also arises as a "stick breaking" process.
Let $F_n(i), i=1,...,k$ be the number of balls of colour $i$ when the urn contains $n$ balls. Then  $S(i)=F_n(i)/n$ tends to the length of the $i$-th fragment when the unit interval is broken into $k$ pieces using $k-1$ independent variates $U$ uniformly distributed in $[0,1]$.
This kind of random partitioning process corresponds to a stick-breaking or spacing process\,, in which a stick is divided into $k+1$ fragments. The distribution of the largest fragment is well-studied \cite{Holst(1980)}, \cite{Pyke(1965)}.


\begin{lemma} (\cite{Holst(1980)}, \cite{Pyke(1965)})  \label{lemma:spacings}
	Suppose a stick of length $1$ is broken into $k$ fragments uniformly at random.
Let $ S_{(1 )}\le  S_{(2)} \le \ldots  S_{(k)}$ be the size of these fragments given in increasing order of size. Then, for $i \in \{1, \ldots , k\}$\,,
	\[ \E  S_{(i)} = \frac{1}{k} \sum_{j=0}^{i-1} \frac{1}{k-j} \,.\]
Thus, the largest fragment has size $\E S_{(k)}=H_k \sim \frac 1k \log k$.
\\
The distribution of  $S_{(k)}$ satisfies
\begin{equation} \label{S(k)}
\Pr(S_{(k)} \ge x) \sim \exp \left( -ke^{-kx} \right),
\end{equation}
\end{lemma}
Thus $\E F_{(1)}$, the expected size of the largest fragment  among $k$
  tends to $\E S_{(k)}= \frac nk \log k$.

\begin{figure}[t]\label{fig:num-comparison1}
	\centering
\includegraphics[scale=.45]{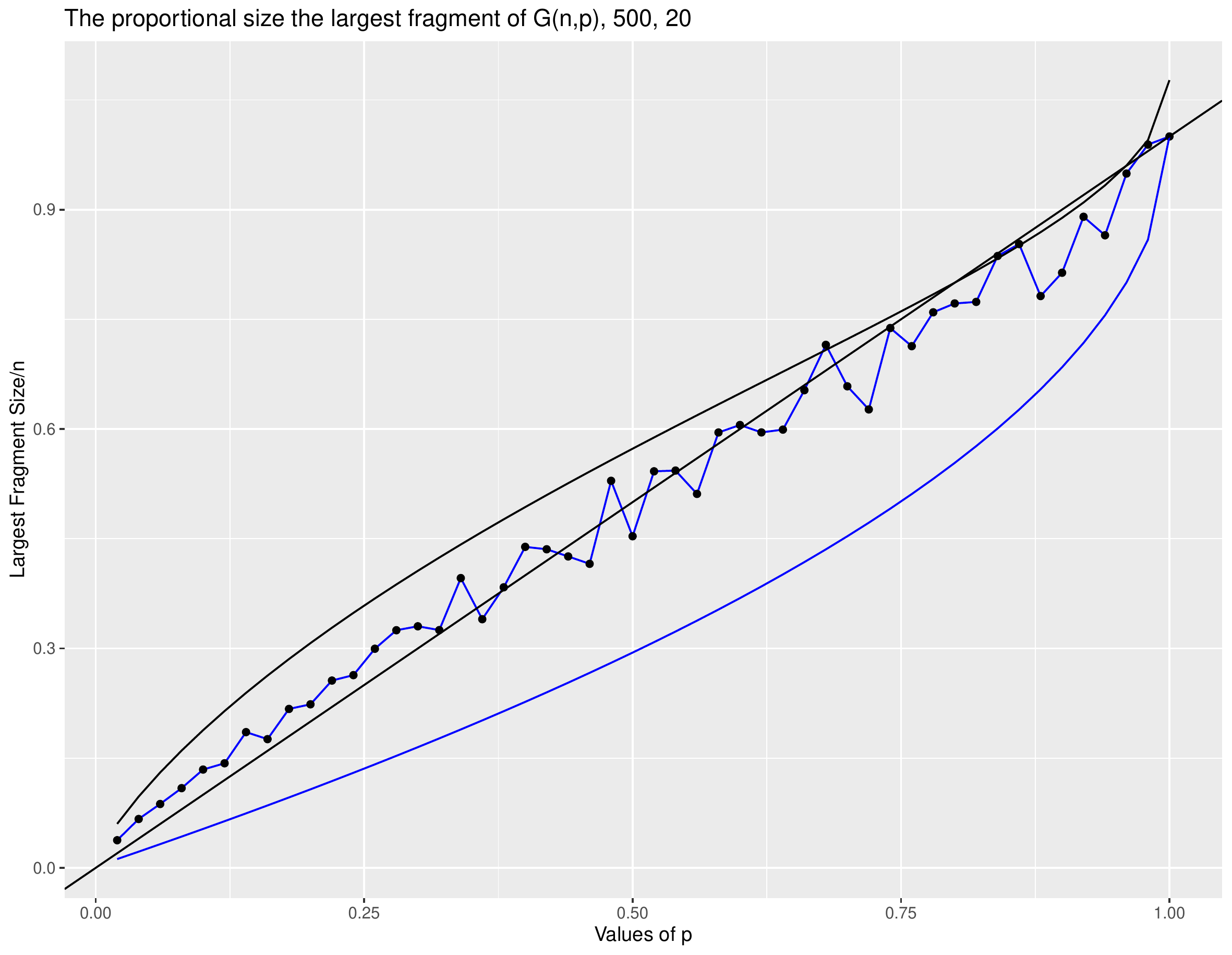}
	\caption{The blue data plot is the average largest fragment size obtained by simulation using $G(n,p)$. For large $m=Np$, $G(n,p) \sim G(n,m)$ so the results are equivalent. The upper line is Lemma \ref{lemma:spacings} for the largest fragment $S_{(k)}$, where $k$ is based on an estimate of $\E a(m)$ using \eqref{Gnm}. The diagonal  black line is $y=x$ for comparison. The lower line is $1/k$, the average component size.
 The plot is for 20 replications at $n=500$, values of $p$ interpolated at $0.02$.}
\end{figure}

\paragraph{ Maximum fragment size. Finite case.}
Lemma \ref{lemma:spacings} although elegant is a limiting result. We check the veracity of the tail distribution   of the maximum fragment size \eqref{S(k)} for finite $n$. It turns out to be quite a lot of work.
The value of \eqref{S(k)}  evaluated at $x=(1/k)(\log k +\log \om)$ is to be compared with  Lemma \ref{FiniteL}.

\begin{lemma}\label{FiniteL}
For $k$  sufficiently large,  $\Pr(F_{(1)} \ge \frac nk (\log k +\om)) =O(e^{-\om})$.\\
If $k\ge 2$ is finite,   the above becomes $\Pr(F_{(1)} \ge \frac n{k-1} (\log k +\om)) =O( e^{-\om})$.
\end{lemma}

\begin{proof}
Recall that $N(k)={n-1 \choose k-1}$ is the number of partitions of $s=n-k$ unlabelled vertices among $k$ distinguishable root vertices (the influencers). Let $N(\ell,k)$ be the number of these partitions which contain at least one fragment of size $\ell$; thus consisting of a root and $\ell-1$ follower vertices.
Using the 'stars and bars' notation given above, there are $k$ ways to choose a left hand bar (a cell) to which we allocate $\ell-1$ stars. There remain $s'=s-(\ell-1)=n-k-\ell +1$ stars to be allocated. Contract the specified cell (stars and delimiters) to a single delimiter. The number of delimiters is now $k'=k-1$, and $n'=s'+k'=n-\ell$. The remaining cells can be filled in ${n'-1 \choose k'-1}$ ways, and thus
\[
N(\ell,k)= k {n - \ell-1 \choose k-2}.
\]
Assume $k \ge 2$ so that $\ell \le n-1$. Let $P_k(\ell)$ be the proportion of partitions which contain {\em at least one fragment} of size $\ell$.
Then
\begin{align*}
P_k(\ell)=& \frac{N(\ell,k)}{N(k)}=\frac{N(\ell,k)}{{n-1 \choose k-1}}\\
=& k(k-1) \frac{(n-k) \dots (n-k-\ell+2)}{(n-1) \cdots (n-\ell+1)(n-\ell)}\\
=& \frac{k(k-1)}{n-\ell} \brac{1-\frac{k-1}{n-1}}\cdots \brac{1-\frac{k-1}{n-\ell+1}}.
\end{align*}
{\bf Case $k$ tends to infinity.} Suppose $k \ge \om$. For any value of $k$, the expected length of a fragment is $n/k$ so
\[
\Pr(\ell \ge \sqrt k \E \ell) \le \frac{1}{\sqrt k}.
\]
Assume $\ell \le \sqrt k \E \ell= n/\sqrt k$. We continue with the asymptotics of $P_k(\ell)$.
\begin{align*}
P_k(\ell) = &\;
\frac{k(k-1)}{n-\ell}
\; \exp\brac{ -(k-1) \brac{\frac{1}{n-1}+ \cdots+ \frac{1}{n-\ell+1}}} \\
&\; \times
\exp\brac{ -(k-1)O \brac{\frac{1}{n^2}+\cdots + \frac{1}{(n-\ell)^2}}}
\\
\sim &\; \frac{k(k-1)}{n-\ell}\; \exp \brac{-(k-1) \log \frac{n-1}{n-\ell}} \;
\exp\brac{-O(k \ell/n^2)}\\
=&\; \frac{k(k-1)}{n-\ell} \bfrac{n-\ell}{ n-1}^{k-1} e^{-O(\sqrt{k}/n)}
\;\;\sim  \;\; k^2 \frac{n-1}{(n-\ell)^2} \brac{1-\frac{\ell+1}{n-1}}^k\\
\sim&\;  \frac{k^2}{n} \exp\brac{-\frac{k(\ell+1)}{n-1}} e^{-O(k(\ell/n)^2)} \;\; \sim \;\;\frac{k^2}{n} \exp \brac{-\frac{k\ell}{n}}.
\end{align*}
Let
\[
P_k( \ell \ge \ell_0)= P_k(\ell_0 \le \ell < n/\sqrt k )+ P_k(\ell \ge n/\sqrt k)=\sum_{\ell=\ell_0}^{n/\sqrt{k}} P_k(\ell) +O(1/\sqrt k).
\]
Then
\[
P(\ell \ge \ell_0) \sim  \frac{k^2}{n} \sum_{\ell=\ell_0}^{n/\sqrt k} e^{-\ell \frac{k}{n}}\;\;
\sim \;\;e^{- \ell_0 k/n} \;\; \frac{k^2}{n}\; \frac{1-e^{-O(\sqrt k)}}{1-e^{-k/n}}\;\;
\sim \;\; k e^{- \ell_0 k/n}.
\]
Let $\ell_0= (n/k) (\log k+\om)$, where $\ell_0 \ll n/\sqrt k$ then
\begin{equation} \label{Pk1}
P_k(\ell \ge \frac{n}{k}(\log k+ \om)) \sim e^{-\om}+ O(1/\sqrt k).
\end{equation}
Thus segments of order $(n/k)(\log k)$ exists with constant probability provided $\om$ is constant. This gives the order of the maximum segment length.

{\bf Case: $k \ge 2$ finite or tending slowly to infinity.}
Returning to a previous expression
\[
P_k(\ell) \sim  \frac{k(k-1)}{n-\ell} \bfrac{n-\ell}{ n-1}^{k-1}.
\]
Put $\ell/n=x$ then returning to the expansion of $P_k(\ell)$, and assuming $k \ge 3$,
\begin{align*}
P_k(\ell \ge \ell_0) & \sim \frac{k(k-1)}{n-1} \sum_{\ell \ge \ell_0} \bfrac{n-\ell}{n-1}^{k-2}\\
\sim & \frac{k(k-1)}{n-1} \int_{\ell_0/n}^1 (1-x)^{k-2} dx = k (1-\ell_0/n)^{k-1}.
\end{align*}
This is similar to the previous case.
\end{proof}

\end{document}